\documentclass[12pt]{article}
\usepackage[T2A]{fontenc}
\usepackage[cp1251]{inputenc}
\usepackage{hyperref}
\usepackage{a4wide}
\usepackage{amsmath,latexsym,amsthm,amsfonts}
\usepackage{amssymb,amscd}
\usepackage{graphicx}

\newcommand{\N}{\Bbb N}

\newcommand{\R}{\Bbb R}

\newcommand{\X}{\frak X}
\newcommand{\bml}[1]{\begin{multline}\label{#1}}
\newcommand{\bee}{\begin{equation}}
\newcommand{\bed}{\begin{displaymath}}
\newcommand{\ee}{\end{equation}}

\newcommand{\bs}{\begin{split}}

\newcommand{\be}{\beta}
\newcommand{\ga}{\gamma} \newcommand{\Ga}{\Gamma}
 \newcommand{\De}{\Delta}
\newcommand{\la}{\lambda} \newcommand{\La}{\Lambda}

\newcommand{\si}{\sigma}

 \newcommand{\sm}{\setminus}

\newcommand{\w}{\widetilde}

\newcommand{\ov}{\overline}

\newtheorem{theorem}{Theorem}[section]
\newtheorem{lem}{Lemma}[section]

\newtheorem{prop}{Proposition}[section]
\newtheorem{remark}{Remark}[section]
\newtheorem{corollary}{Corollary}[section]
\theoremstyle{definition} \newtheorem{df}{Definition}[section]

\theoremstyle{remark} \newtheorem{rem}{Remark}[section]
 \numberwithin{equation}{section}

\begin{document}

\title{ On  quasi-continuous  approximation in classical  statistical mechanics}
\author{S.M.~Petrenko$^1$, A.L.~Rebenko$^2$, M.V.~Tertychnyi $^2$}
\date{}
\maketitle
\begin{footnotesize}
\begin{tabbing}
$^1$ \= Lvivs'kyi NLTU, Lviv, UKRAINE\\
\> {\em petrenko2003kiev@inbox.ru}\\
$^2$\= Institute of Mathematics, Ukrainian National Academy of Sciences, Kyiv, Ukraine \\
\>{\em rebenko@voliacable.com } ;{ \em maksym.tertychnyi@gmail.com}\\

\end{tabbing}
\end{footnotesize}
\begin{abstract}
A continuous infinite system of point particles with strong
superstable interaction is considered in the framework of
classical statistical mechanics. The family of approximated
correlation functions is determined  in such a way, that they take
into account only such configurations of particles in
$\mathbb{R}^d$ which for a given partition of the configuration
space $\mathbb{R}^d$ into nonintersecting hyper cubes with a
volume $a^d$ contain no more than one particle in every cube. We
prove that these functions converge to the proper correlation
functions of the initial system if the parameter of approximation
$a\rightarrow 0$ for any positive values of an inverse temperature
$\beta$ and a fugacity  $z$. This result is proven both for
two-body interaction potentials and for many-body case.
\end{abstract}

\thispagestyle{empty}



\section{\bf Introduction}

 The procedure of lattice approximation
is very often used to study continuous systems. There is a
well-known example of the lattice approximation in Euclidean
quantum field theory for the model $\la:\phi^4:$ in the
two-dimensional space-time which transforms the system to Ising
model with unbounded continuous spin. In contrast to Euclidean
quantum field theory, where lattice systems play the role of
approximation, in statistical mechanics they represent part of the
Nature, such as ferromagnetics, quantum oscillators etc. The
theory of such systems is well developed, unlike continuous
systems such as dense gases and liquids. The main difficulties in
the mathematical description of continuous systems in statistical
mechanics are accumulation of many number of particles in small
volumes. To avoid this problems such systems as {\it lattice gas}
were invented to describe some general characteristics of real
gases. But in majority of works there was no parameter which in
some sense restored systems to continuous gases.

In this work we propose some intermediate approximation of
continuous gases, which is very close to lattice gases and all
main characteristics of continuous gases can be obtained with help
of limit transition.

 Quasi-continuous  approximation  of the
Equilibrium Classical Statistical Mechanics was proposed in the
article \cite{RT07} for the investigation of infinite systems of
interacting point particles with two-body strong superstable
potentials. The matter of this approximation is that in integrals
which are in the definitions of the main characteristics such as
partition function and correlation functions integrations are
realized over such configuration which for a given partition of
the configuration space $\mathbb{R}^d$ into nonintersecting hyper
cubes with a volume $a^d$ contain no more than one point in every
cube. Correlation functions and pressure of systems defined in
such a way though have a proper limit at $a\rightarrow 0$ even for
infinite volume systems if the interaction potential is
sufficiently singular at the origin, more exactly if the potential
is locally nonintegrable in any bounded region of $\mathbb{R}^d$
which contains an origin. This fact though is predictable from the
physical point of view but from mathematical point of view it is a
little bit unexpected as the Poisson measure (and Gibbs measure
too) of the set of such configurations is zero.

At the same time, such defined system can be approximated by the
 {\it lattice gas}, an investigation of which is
considerably simplified. This transition from continuous to
lattice systems and vice versa  is particularly important in the
investigation of critical behavior of infinite systems near phase
transition points.

 It was proved in the article \cite{RT07} that for any
positive values of temperature $T$ (or inverse temperature
$\be=1/kT$) and fugacity $z$ of infinite classical systems the
approximated pressure $p^{(-)}(z,\be;a)$, where $a$ is the
parameter of approximation, tends to the proper value of the
pressure $p(z,\be)$ of the considered statistical system as
$a\rightarrow 0$. In the article \cite{Pe08} this result was
generalized for the systems with many-body interactions. Later, in
the article \cite{RT09} the same result was obtained for  family
of the correlation functions, but only for sufficient small values
of fugacity $z$, the values of which were bounded by the radius of
convergence of the Kirkwood-Salsburg expansion for the correlation
functions.

 In this article we are going to generalize this result for the case
 of arbitrary positive values of fugacity $z$ and temperature $T$. Using an
 expansion in so-called {\it dense configurations}, which was
 proposed in \cite{Re98} for finite range interaction and
 in the article \cite{PR07} for infinite range potentials, we prove that the family of
 approximated correlation functions $\rho_\La^{(-)}(z,\be;a)$ of the finite volume
 ($\La\Subset\R^d$)  are uniformly bounded by
 a constant which does not depend on the parameter of approximation
 $a$ and volume $\La$ and have pointwise limit $\rho(z,\be)$
as $\La\uparrow\R^d$ and $a\rightarrow 0$ for arbitrary values of
fugacity $z$ and temperature $T$. This result will be proved both
for two-body interaction potentials and for many-body potentials
of general superstable type.

\section{\bf Configuration spaces}

\subsection{\bf The main configuration spaces}
 Let ${{\Bbb R}}^{d}$ be a $d$-dimensional Euclidean
space.  The set of positions $\{x_i\}_{i\in{\Bbb N}} $ of
identical particles is considered to be  a locally finite subset
in ${\Bbb R}^d$ and the set of all such subsets creates the
configuration space:
\begin{equation*}
\Gamma=\Gamma_{{\Bbb R}^{d}} :=\left\{ \left. \gamma \subset
{{\Bbb R}}^{d}\right| \,|\gamma \cap \Lambda
|<\infty,\,\,\mathrm{for}\,\,\mathrm{all}\; \Lambda \in
\mathcal{B}_c({{\Bbb R}}^{d})\right\},
\end{equation*}
where $|A|$ denotes the cardinality of the set $A$ and
$\mathcal{B}_{c}({{\Bbb R}}^{d})$ denote the systems of all
bounded Borel sets in ${\Bbb R}^d$. We also need to define the
space of finite configurations $\Gamma_{0}$:
\[
\Gamma_{0} = \bigsqcup_{n\in {\Bbb N}_0}\Gamma^{(n)},\quad
\Gamma^{(n)}:= \{\eta\subset{\Bbb R}^{d}\; |\; |\eta|=n,\;n\in
{\Bbb N}_0 \},\quad {\Bbb N}_0={\Bbb N}\cup \{0\}.
\]

 For every $\Lambda \in \mathcal{B}_c({{\Bbb R}}^{d} )$ one
can define a mapping $N_\Lambda :\Gamma\rightarrow \Bbb{N}_0$ of
the form
\[
N_\Lambda (\eta ):=|\eta \cap \Lambda |.
\]
 The Borel $\sigma
$-algebra $ {\frak B}(\Gamma)$ is equal to $\sigma (N_\Lambda
\left| \Lambda \in \mathcal{B}_c({{\Bbb R}}^{d})\right.)$.
 See \cite{Le75I},
\cite{Le75II} for details.

We need also to define

\begin{equation*}
\Ga_\La:=\left\{ \left. \gamma \in\Ga_0 \right| \,\gamma \subset
\Lambda ,\,\,\; \Lambda \in
\mathcal{B}_c({\mathbb{R}}^{d})\right\},
\end{equation*}

 By ${\frak B}(\Gamma_\Lambda)$ we denote the corresponding
$\sigma$-algebras on $\Gamma_\Lambda$ and $\Gamma_{0,\Lambda}$.

\subsection{\bf Lebesgue-Poisson measure }

Let $\sigma$  be Lebesgue measure on $\mathcal{B}({\Bbb R}^d)$ and
for any $n\in \Bbb{N}$ the product measure $\sigma ^{\otimes n}$
can be considered as a measure on
$$\widetilde{({{\Bbb R}}^{d})^n}=
\left\{ \left. (x_1,\ldots ,x_n)\in ({{\Bbb R}}^{d})^n\right| \,x_k\neq x_l\,\,\mathrm{if}%
\,\,k\neq l\right\}$$ and hence as a measure $\sigma ^{(n)}$ on
$\Gamma^{(n)}$ through the map
$$\:sym_n:\widetilde{({\Bbb R}^{d})^{n}}\ni
(x_1,...,x_n)\mapsto\{x_1,...,x_n\}\in\Gamma^{(n)}.$$

Define the Lebesgue-Poisson measure $\lambda_{z \sigma}$ on
${\frak B}(\Gamma_{0})$ by the formula:
\begin{equation}\label{1}
 \lambda_{z\sigma} :=\sum_{n\ge 0}\frac{z^n}{n!}\sigma^{(n)}.
\end{equation}

The restriction of $\lambda_\sigma$ to ${\frak B}(\Gamma_\Lambda)$
we also denote by $\lambda_\sigma$. For more detailed structure
and analysis of the configuration spaces $\Gamma$, $\Gamma_0$,
$\Gamma_\Lambda$  see \cite{AKR97}.

\subsection{\bf Partition of $\R^d$ }

 Following Ruelle \cite{Ru70} define the partition of the Euclidean
space $\R^d$ into elementary cubs. Let $a>0$ be arbitrary. For
each $r\in{\Bbb Z}^{d}$ we define an elementary cube with an edge
$a$ and a center $a r$:
\begin{equation}\label{2-1}
\Delta_{a}(r):=\{x\in{\Bbb R}^d\mid a(r^i-1/2)\leq
x^i<a(r^i+1/2)\}.
\end{equation}
 We will  write $\Delta$ instead of $\Delta_{a}(r)$,
if a cube $\Delta$ is considered to be arbitrary and there is no
reason to emphasize that it is centered at the concrete point
$ar$. Let $\overline\Delta_a$ be the partition of $\mathbb{R}^d$
into cubes $\Delta_{a}(r)$. Define, also,  the notion of {\it
compatible partitions}.
\begin{df}\label{d:0}
Two partitions $\overline\Delta_a$ and $\overline\Delta_{a'}$ with
$a'<a$ are compatible if $a/a'\in \N$ and partition
$\overline\Delta_a$ can be obtained from the partition
$\overline\Delta_{a'}$ removing all edges of its cubes which do
not lie on the edges of the partition $\overline\Delta_a$.
\end{df}

To avoid some confusion we work in this article only with
compatible partitions.

\subsection{\bf Additional configuration spaces }

Define two additional configuration spaces:
$\Gamma_{\Lambda}^{dil}$ we call a space of {\it dilute}
configurations and $\Gamma_{\Lambda}^{den}$  a space of {\it dense
} configurations.

 Without any restriction
of general case, we consider only that $\Lambda \in
\mathcal{B}_c({{\Bbb R}}^{d})$ which is union of cubes
$\Delta_{a}(r)$ with some fixed $a$, which depends on the
interaction potential.  In the cases where this  particular
partition will be important we denote by $\La(a)$ the union of
such cubs. Then

\begin{equation}\label{2}
\Gamma^{dil}_{\Lambda} :=\left\{ \left. \gamma \in\Gamma_{\Lambda}
\right| \,|\gamma_\Delta|=0 \vee 1 \; \text{for all}\;
\Delta\subset\Lambda \right\}
\end{equation}

and

\begin{equation}\label{3}
\Gamma^{den}_{\Lambda} :=\left\{ \left. \gamma \in\Gamma_{\Lambda}
\right| \,|\gamma_\Delta|\geq2 \; \text{for all}\;
\Delta\subset\Lambda \right\}.
\end{equation}
For any $\Delta\in\overline\Delta_a$ and any fixed configuration
$\eta\in\Gamma_{\Lambda}$ we split the space of {\it dense }
configurations $\Gamma^{den}_{\Delta}$  into two subspaces:
\begin{equation}\label{dec1}
\Gamma_{\Delta}^{(>)}(\eta) =  \Gamma_\Delta^{(>)} :=\left\{
\left. \gamma \in\Gamma^{den}_{\Delta} \right| \, |\gamma|>
d_{\eta}^{\varepsilon}(\Delta)\right\}
\end{equation}
and
\begin{equation}\label{dec2}
\Gamma_\Delta^{(<)}(\eta) = \Gamma_\Delta^{(<)} :=\left\{ \left.
\gamma \in\Gamma^{den}_{\Delta} \right| \, |\gamma |\leq
d_{\eta}^{\varepsilon}(\Delta)\right\},
\end{equation}

where $\Delta\equiv\Delta_{a}(r), \,0<\varepsilon\leq1 $ and
\begin{equation}\label{dec3}
d_{\eta}(\Delta)=dist(\eta,\mathbf{\Delta}),
d_{\eta}^{\varepsilon}(\Delta)=(d_{\eta}(\Delta))^\varepsilon,
\end{equation}
where $\mathbf{\Delta}$ is the closure of the cube $\Delta$. It's
obviously that
$\Gamma_{\Delta}^{den}=\Gamma_\Delta^{(>)}\cup\Gamma_\Delta^{(<)}$.
And finally for $X_k=\cup_{i=1}^k\De_a(r_i)$
\begin{equation}\label{dec4}
\Gamma_{X_k}^{(>)}(\eta) =  \Gamma_{X_k}^{(>)} :=\left\{ \gamma
\subset X_k\mid \, |\gamma_\De|> d_{\eta}^{\varepsilon}(\Delta)\;
\text{for all}\; \Delta\subset X_k \right\}
\end{equation}
and
\begin{equation}\label{dec5}
\Gamma_{X_k}^{(<)}(\eta) =  \Gamma_{X_k}^{(<)} :=\left\{ \gamma
\subset X_k\mid  \, |\gamma_\De|\leq
d_{\eta}^{\varepsilon}(\Delta)\; \text{for all}\; \Delta\subset
X_k \right\}.
\end{equation}

\section{\bf Interaction}

For the general case  interaction between particles is realized by
infinite sequence of interaction potentials:
\begin{equation} \label{ch1-Many-Body_potential}
V=(0,0,V_2(x_1,x_2),V_3(x_1,x_2,x_3),...,V_p(x_1,...,x_p),...)
\end{equation}
In case of two-body interaction, which is the most popular among
physicists  components of the sequence
\eqref{ch1-Many-Body_potential} look like:
\begin{equation}\label{V-2}
V_2(x_1,x_2)=\phi(|x_1-x_2|),\;V_p\equiv 0,\;p\geq 3,
\end{equation}

The energy of any configuration $\gamma\in\Gamma_0$ is defined by
the following formula:
\begin{equation}\label{U}
U(\gamma)=U_V(\gamma)=\sum_{p=2}^{|\gamma|}\sum_{\{x_1,...,x_p\}\subset
\gamma}V_{p}(x_1,...,x_p)=\sum_{\eta\subseteq \gamma :|\eta|\geq
2}V(\eta),
\end{equation}
and interaction energy between two configurations
$\eta,\;\ga\in\Ga_0$ by
\begin{align}\label{EnergyFunctional}
W(\eta;\gamma)=W_V(\eta;\gamma)\;&=\;U(\eta\cup\gamma)-U(\eta)-U(\gamma)=\\
&=\sum_{p=2}^{|\eta\cup\ga|}\sum_{\substack{i,j=1 \\
i+j=p}}^{|\eta|,|\ga|}\sum_{\substack{\{x_1,...,x_i\}\subset\eta \\
\{y_1,...,y_j\}\subset\gamma}}V_p(x_1,...,x_i,y_1,...,y_j).\notag
\end{align}
The correspondent formulas for two-body interaction are:
\begin{equation}\label{UP}
U(\gamma)=U_\phi(\gamma)=\sum_{\{x_1,x_2\}\subset
\gamma}\phi(|x_1-x_2|),
\end{equation}
\begin{equation}\label{EnergyFunctional-V-2}
W(\eta;\gamma)\;=\;W_\phi(\eta;\gamma)\;=\; \sum_{\substack{x\in\eta \\
y\in\gamma}}\phi(|x-y|).
\end{equation}

We introduce 3 kinds of interactions, which will be used in this
article:
\begin{df}\label{d:1}
Interaction  $U$ is called: \\
a)\;stable {\bf (S)}, if there exists $B$>0 such that:\\
\begin{equation}\label{7}
U(\gamma)\geq-B|\gamma|,\; \text{for any\;}\gamma \in \Gamma_{0};
\end{equation}
b)\;superstable {\bf (SS)}, if there exist $A>0, \,B\geq0 $
and partition  $\overline{\Delta_{a}}$ such that:\\
\begin{equation}\label{8}
U(\gamma)\geq A \underset{\Delta\in \overline{\Delta_{a}} }
{\sum}\; |\gamma_{\Delta}|^2 - B|\gamma|,\; \text{for any\;}\gamma
\in \Gamma_{0};
\end{equation}
c)strong superstable {\bf (SSS)}, if there exist  $m \geq 2$,
$a_0>0$ s.t. for any $0< a \leq a_0$ there exist $A(a)>0$,
$B(a)\geq 0$ s.t.
\begin{equation} \label{9}
U(\gamma)\geq A(a) \underset{\Delta\in
\overline{\Delta}_{a}:|\ga_\De|\geq 2 } {\sum}\;
|\gamma_{\Delta}|^m - B(a)|\gamma|,\; \text{for any\;}\;\;\gamma
\in \Gamma_{0}.
\end{equation}
\end{df}
 In accordance
with these definitions   there is a  problem to describe
conditions on   potentials, which ensure stability, superstability
or strong superstability of an infinite statistical system. This
problem has a long story. A short review of this problem and some
new results one can find in  \cite{RT08} and \cite{Te08}.

\begin{rem}
 It is clear   that if the equation \eqref{8} holds for some
 partition $\overline{\Delta_{a}}$ with the constants
  $A$ and $B$ then it holds  with the same constants $A$ and $B$
   for any partition $\overline{\Delta_{a'}}$
  for which $a'<a$ and they are compatible.
\end{rem}

\begin{rem}
 It is clear that if the potential is strong superstable then it is
simply superstable with $A=A(a_0),\;B=B(a_0)$.
\end{rem}

\subsection{\bf Definition of the system with two-body interaction}

{\bf (A): Assumption on the interaction potential.}
\textit{Consider a general type of potentials \,$\phi$,\, which
are continuous on $\mathbb{R_+}\sm\{0\}$  and for which there
exist \;\,$ r_0 > 0,\;\,R\;
>\; r_0, \\
\varphi_{0}>0,\,\varphi_{1}>0,\, \text{and}\;\, \varepsilon_0
> 0$\, such that:}
\begin{align}
&1)\, \phi(|x|)\equiv -\phi^-(|x|)\geq -
\frac{\varphi_{1}}{|x|^{d+\varepsilon_0}}\;\;\;
  \text{for}\;\;\; |x|\geq R,\label{211};\\
&2)\, \phi(|x|)\equiv\phi^+(|x|)\geq
\frac{\varphi_{0}}{|x|^{s}},\, s\geq d\;\;\;
\text{for}\;\;\; |x|\leq r_0,\label{212}
\end{align}
\textit{where}
\begin{equation}\label{213} 
\phi^+(|x|):= \max \{0, \phi(|x|)\}, \, \phi^-(|x|):=-\min \{0,
\phi(|x|)\}.
\end{equation}
Note that  in the definition 3.1, c)(SSS) the constant $a_0\leq
r_0$. For the interaction potentials which satisfy the assumptions
{\bf (A)} define two important characteristics (for any
$\Delta\in\overline\Delta_a$ with $a \leq a_0$ ):
\begin{align}\label{219}
&1)\quad
\upsilon_\varepsilon(a):=\sum_{\Delta^{'}\in\overline\Delta}
\,\sup_{x\in\Delta}\,\sup_{y\in\Delta^{'}}\phi^-(|x-y|){|x-y|}^\varepsilon,\,\,\text
{ for any }\varepsilon <\varepsilon_0;\\
&2)\quad b(a):=\inf_{\{x,y\}\subset\Delta}\phi^+(|x-y|).\label{220}
\end{align}
Due to the translation invariance of the 2-body potential the
values $\upsilon_0 $ and $b$ do not depend on the position of
$\Delta$. The following statement is true.

\begin{prop}\label{prop2.1}
Let potential $\phi$ satisfy the assumption {\bf (A)}. Then the
interaction is strong superstable and the energy $U$ satisfies the
inequality \eqref{9} with some $0<a_0<r_0$ and if $s>d$ then
\begin{equation}\label{218}
m=2,\;\; A(a)\;=\;\frac{b(a)-2\upsilon_0(a)}{4} >
0,\;\;B(a)\;=\;\frac{\upsilon_0(a)}{2}
\end{equation}
for $a\leq a_0$.
\end{prop}
See the proof  in \cite{RT09}. More powerful result was obtained
in the article \cite{RT08}, but for our goals it is sufficient to
apply  the inequalities \eqref{218}.

Following \cite{PR07} we introduce the following notations, which
will be used in our future estimates:
\begin{align}
&\phi_{\delta}^+ (|x|):=(1-\delta)\,\phi^+(|x|),\,\,\,\,\,
U_{\delta}^+:=U_{\phi_{\sigma}^+},\label{13}\\
&\phi_\delta^{st}:=\delta\phi^+(|x|)-\phi^-(|x|),\,\,\,\,\,U_\delta^{st}:=U_{\phi_\delta^{st}},\,\,\,\delta\in(0,1).\label{14}
\end{align}
One can deduce from \eqref{13}, \eqref{14}, that:
\begin{equation}\label{15}
\phi(|x|)\,=\,\phi_{\delta}^+(|x|)+\phi_{\delta}^{st}(|x|),\,\,\,\,U(\gamma)\,=\,U_{\delta}^+
(\gamma)+U_{\delta}^{st} (\gamma).
\end{equation}

\begin{prop}\label{prop2.2}
Let potential $\phi$ satisfy the assumption {\bf (A)}. Then there
exist   $0<a_*<r_0$ such that for any constant $ \delta\in(0,1/2)
$
\begin{equation}\label{21}
(1-\delta) b(a) > 2\upsilon_0(a),\;\;\text{for}\;\; a\leq a_*
\end{equation}
and the potential $ \phi_\delta^{st}$  is stable: $
U_\delta^{st}:=U_{\phi_\delta^{st}}(\gamma)\,\geq\,-B_\delta
|\gamma|\,,\,\,\,\gamma\in\Gamma_0 $ with
\begin{equation}\label{21-1}
B_\delta\;=\;\frac{1}{2}\upsilon_0(a_*)\;=\;\frac{\delta}{4}b(a_*).
\end{equation}
\end{prop}

{\it Proof.} The inequality \eqref{21} follows from the assumption
{\bf (A)} and the definitions \eqref{219} and \eqref{220} as for
small $a$ they behave as:
\begin{equation}\label{21-2}
 b(a)\sim
 \frac{\varphi_0}{a^s}\;\;\text{and}\;\;\upsilon_\varepsilon(a)\sim\frac{\phi_\varepsilon}{a^d},
\end{equation}
and for $s>d$ we can choose sufficiently small $a=a_*$ or
$\varphi_0>>\phi_\varepsilon$ for $s=d$, where
\begin{equation}\label{21-3}
 \phi_\varepsilon\;=\;\int_{\R^d}\varphi^-(|x|)|x|^\varepsilon dx.
\end{equation}
As in \cite{RT09} (see Proposition 2.1) one can calculate that
\begin{equation}\label{21-4}
U_{\varphi_\delta^{st}}(\gamma)\;\geq\;\sum_{\Delta\in\,\overline\Delta_a:|\ga_\De|\geq
2} |\gamma_\Delta|^2\left(\delta\frac{
b(a)}{4}-\frac{\upsilon_0(a)}{2}\right)-\frac{\upsilon_0(a)}{2}|\gamma|.
\end{equation}
Let us chose  $a_*$ as a root of equation
\begin{equation}\label{21-5}
\delta\frac{ b(a)}{4}-\frac{\upsilon_0(a)}{2}=0.
\end{equation}
Then to satisfy \eqref{21} we have to choose $\delta >1/2$ and the
constant $B_\delta$ in \eqref{21-1}  can be expressed in terms of
parameters of the interaction potential $\varphi_0,\;\phi_0,\;s$
and dimension of the space $d$ (see Proposition 2.2 in
\cite{RT09}).
 \hfill $\blacksquare$

\subsection{\bf Definition of the system with many-body interaction}

In this section we consider a general type of many-body
interaction specified by a family of $p$-body potentials
$V_{p}:\R^{dp}\rightarrow\R,\:p\geq 2$. About the family of
potentials $V:=\{V_{p}\}_{p\geq 2}$ we will assume:

{\bf A1. Continuity.}
\begin{displaymath}
V_p\in C(\widetilde{(\R^{d})^{p}}),\;p\geq2,
\end{displaymath}
 {\it where}
\[
\widetilde{({\R}^{d})^{\otimes n}}= \left\{ \left. (x_1,\ldots
,x_n)\in ({\R}^{d})^{\otimes n}\right| \,x_k\neq x_l\,\,\text
{при} \,\,k\neq l\right\}.
\]

\vspace{0,5cm}

{\bf A2. Symmetry.} For any $p\geq2$, any
$(x)_p)=(x_1,...,x_p)\in(\R^{d})^{p}$, and any permutation $\pi$
of numbers $\{1,\ldots,p\}$
\begin{displaymath}
V_p(x_1,...,x_p)=  V_p(x_{\pi(1)},...,x_{\pi(p)}).
\end{displaymath}

{\bf A3. Translation invariance.} For any $p\geq2$, any
$(x_1,...,x_p)\in(\R^{d})^{p}$, and any $x_0\in\R^d$
\begin{displaymath}
V_p(x_1,...,x_p)=  V_p(x_1+x_0,...,x_p+x_0).
\end{displaymath}

{\bf A4. Superstability.} For any $p\geq2$ the potentials $V_{p}$
can be represented as
\begin{align}\label{V-pV-st}
&V_p\;=\;  \w{V}_p^+
+V_p^{(st)},\;\;\;V_p^{(st)}\;=\;\overline{V}^+_p\;+\;V_p^-,\\
&\w{V}^+\;:=\;(\w{V}^+_p)_{p\geq
2}\;\;\;V^{(st)}\;:=\;(V^{(st)}_p)_{p\geq 2},\notag
\end{align}

where $\w{V}_p^+\;+\;\overline{V}^+_p=V_p^+$, $V_p^{\pm}$ are
defined in the same way as in \eqref{213} and $V_p^{(st)},\;p\geq
2$ provides the stability of the corresponding energy $U$ , i.e.
there exists a constant $B\geq 0$ such that for any configuration
$\eta\in\Ga_{0}$
\begin{align}\label{(S)}
U_{V^{(st)}}(\eta)\geq -B|\eta|.
\end{align}

The corresponding decomposition for the energy:
\begin{equation}\label{63}
U(\gamma)\,=\,U^+ (\gamma) + U^{st} (\gamma).
\end{equation}
Sufficient conditions on the potentials $V_{p}$ providing
superstability inequality  were obtained in \cite{Te08}.

In the article \cite{Ru70} uniform (in volumes $\La_n$) bounds for
the family of correlation functions were obtained for potentials
which guarantee {\it superstability }(SS) and {\it low regularity
condition }(LR) (see \cite{Ru70}). For 2-body potentials which
satisfy the assumptions {\bf (A)}  both of these conditions are
fulfilled. But for many-body potentials which are not positive for
$p\geq 3$  LR-condition is not satisfied. So, as in the articles
\cite{KR04} and \cite{PR09}we formulate so called {\it
attraction-repulsion relations}(instead of LR-condition ) which
gives a possibility to obtain uniform bounds.

To formulate these assumption for potentials $V_p$, consider some
auxiliary constructions.
 Let $p\geq2$ and $N\in\mathbb{N}$.
For any union $X_{N}:=\cup_{j=1}^{N}\De_{j}$ of cubes $\De$ from
the partition $\overline\De_a$ (див. \eqref{2-1}) and any
 $\varepsilon\geq 0$ define  values:

\begin{equation}\label{64}
I_{p}^{k_1,...,k_N}\left(\De_1;...;\De_N\right):=
\sup_{\stackrel{x_{i_1}^{(1)}\in\De_1,...,
x_{i_N}^{(N)}\in\De_N}{i_1=\overline{1,k_1},...,i_N=\overline{1,k_N}}}
V_p^{-}(x_1^{(1)},...,x_{k_N}^{(N)}),
\end{equation}
where $k_1+\cdots+k_N=p, k_j\geq 1, j=\overline{1,N}$ and
\begin{align}\label{h9}
&I_{p}^{k_1,...,k_M|\bar{k}}\left(\De_1;...;\De_M|\varepsilon;(\De)_\pi\right):=\\
&=\sum_{\De'_{1},...,\De'_{\bar{k}}\in\overline{\De}_a}
I_{p}^{k_1,...,k_M,1,1,...,1}\left(\De_1;...;\De_M;\De'_1;...;\De'_{\bar{k}}\right)
\prod_{i=1}^{\bar{k}}\left(1+d^{\varepsilon}_{\De'_i,\De_{\pi(i)}}\right),\notag
\end{align}
where
$d^{\varepsilon}_{\De'_i,\De_{\pi(i)}}=\left(dist({\mathbf\De}'_i,{\mathbf\De}_{\pi(i)})\right)^\varepsilon$,
$\pi$ is the mapping of indices  $\{1,...,\bar{k}\}$ into the set
of indices $\{1,...,M\}$,
$(\De)_\pi\,:=\,\{\De_{\pi(1)},...,\De_{\pi(\bar{k})}\}$\; і\;
$k_1+...+k_M+\overline k=p$. The distance between cubes is the
distance between their closures.

Note that because of translation invariance of interaction
potentials  for $M=1$ all indices $\pi(i)=1$ and

\begin{equation}\label{h9-1}
I_{p}^{k_1|\bar{k}}\left(\De_1|\varepsilon;\De_1\right)\;=
\;I_{p}^{k_1|\bar{k}}\left(a;\varepsilon\right),
\end{equation}
i.e.  it depends on the size of cube $\De_1$, but it does not
depend on positions of $\De_1$. For a   positive part
 $\w{V}_p^+$ of interaction potentials define the following values:

\begin{equation}\label{f3}
v_p^{k_1,\ldots,k_N}\Big(\De_1,...,\De_N\Big):=\inf_{\stackrel{x_{i_1}^{(1)}\in\De_1,...,
x_{i_N}^{(N)}\in\De_N}{i_1=\overline{1,k_1},...,i_N=\overline{1,k_N}}}
\w{V}_p^{+}(x^{(1)}_1,...,x^{(N)}_{k_N}).
\end{equation}

\vspace{0,5cm}

 ${\bf A5}.${\bf Attraction-repulsion relations.} {\it There exist
$a_0>0\:$, such that for any  $N\in\mathbb{N}$, any set
$X_{N}:=\cup_{j=1}^{N}\De_{j}, \Delta_j\in \overline\Delta_{a}$
with $a\leq a_0$ the following inequalities are true:

(i)  for any $\De\in\overline{\De}_{a}$ and any $p\geq2$
\begin{equation}\label{v>0}
V_p(x_1,...,x_p)\geq0, \,\,\,\,\text{if}\;\;\; \{x_1,...,x_p\}
\subset \Delta.
\end{equation}

(ii) for any } $ p \geq 2,\, 1\leq N< p$, \; and $\pi:
\{1,...,n\}\mapsto\{1,...,N\} $

\begin{align}\label{f5}
&v_p^{k_1,...,k_N}(\De_1,...,\De_N)\geq\\
&2\sum_{l=0}^\infty\sum_{\stackrel{m_i\geq 1,i=\ov{1,N};n\geq
1}{m_1+\cdots+m_N+n=p+l}}C_{k_1}^{m_1}\cdots C_{k_N}^{m_N}(2p)^n
I_{p+l}^{m_{1},\ldots,m_{N}|n}(\De_1,...,\De_N;\varepsilon,(\De)_\pi),\notag
\end{align}
where $k_1+\cdots+k_N=p,\; C_k^m\,=\,k!/m!(k-m)!$, if $k\geq m$ і
$C_k^m\,=\,0$ if $m > k$.

\vspace{0,5cm}

\begin{remark}\label{r3-00}
{\it Inequality \eqref{f5} is a consequence of the combinatorial
arguments, which is relevant to control  the negative part of
interaction potentials. From the physical point of view it means
that for the case when there are at least two particles in some
cube (just this situation takes place in case  $N<p$), then for
sufficiently  small size of a cube edge their $p$-body repulsion
energy has to be greater than the attraction energy of these two
particles for all particles of a system and for all $l\geq p$-body
interactions.}
\end{remark}

\begin{lem}\label{2.1}
Let the sequence of potentials $V=\{V_p\}_{p\geq 2}$ satisfy
${\bf A1-A5}$. Then the interaction is strong superstable {\bf
(SSS)}, i.e. there exist  $m \geq 2$, $a_0>0$ s.t. for any $0< a
\leq a_0$ there exist $A(a)>0$, $B(a)\geq 0$ s.t.
\begin{equation} \label{9-1}
U(\gamma)\geq A(a) \underset{\Delta\in
\overline{\Delta}_{a}:|\ga_\De|\geq 2 } {\sum}\;
|\gamma_{\Delta}|^m - B(a)|\gamma|,\; \text{for any\;}\;\;\gamma
\in \Gamma_{0}.
\end{equation}
with
\[
A(a)=v_2^2(a)- 2\sum_{p\geq 2}4^pI_p^{1|p-1}(a;0),\;\;
B(a)=\sum_{p\geq 2}I_p^{1|p-1}(a;0), m=2,
\]
 and for any $\ga\in\eta\cup\Ga_{X'}^{(>)}$ and
$\overline\ga\in\Ga_{X}^{(<)}\cup\Ga_{\Lambda\sm (X\cup
X')}^{(dil)}$, \; $X'\cap X=\emptyset$,
\begin{equation}\label{lem3-1}
 -\beta W(\gamma |\overline{\gamma}) - \frac{1}{2} \beta
U_{\w{V}^+}(\gamma)\,\leq\,\beta\bar{I} |\eta|,
\end{equation}
where $\bar I(a) :=\sum_{p\geq 2}2^p I_{p}^{1|p-1}(a,0)$ {\rm (see
\eqref{h9}-\eqref{h9-1})}.
\end{lem}

{\it Proof.} The main line of the proof is the same as the proof
of Lemma 3.2 in the article \cite{KR04} and  as the proof of Lemma
3.1 in the article \cite{PR09}. The main difference is in the fact
that for obtaining the inequality \eqref{lem3-1} we use a little
bit cumbersome but weaker condition \eqref{f5} than in
\cite{KR04}, \cite{PR09}. \hfill $\blacksquare$

\subsection{\bf Partition functions  and corellation functions}
We introduce an important function, which will be used for the
approximation of statistical systems:
\begin{equation}\label{22}
  \chi_-^{\Delta}(\gamma) \;
 = \left\{ \begin{array}{ll}
    1,  & \mbox{for $\gamma$ with $|\gamma_{\Delta}|=0\vee 1,$} \\
    0,  & \mbox{ otherwise}.
           \end{array}
   \right.
\end{equation}
Let us write an expression for the statistical sum, that includes
all possible configurations from $\Gamma_{\Lambda}$ and  an
expression for the statistical sum, that includes only dilute
configurations from $\Gamma_{\Lambda}^{dil}$:
\begin{align}
& Z_\Lambda(z, \beta ):=\int_{\Gamma_\Lambda}e^{-\beta
U(\gamma)}\lambda_{z\sigma}(d\gamma)\label{23},\\
& Z_\Lambda^{(-)}(z, \beta, a):=\; \int_{\Gamma_\Lambda}e^{-\beta
U(\gamma)}\prod_{\Delta\in\overline{\Delta_{a}}\cap\Lambda}\chi^\Delta_-(\gamma)\lambda_{z\sigma}(d\gamma)
:=\int_{\Gamma_\Lambda}e^{-\beta
U(\gamma)}\lambda^a_{z\sigma}(d\gamma).\label{24}
\end{align}

Let us define correlation function $\rho_{\Lambda}(\eta; z,
\beta)$
 in the case of grand canonical ensemble:
\begin{equation} \label{27}
\rho_\La(\eta; z, \beta) :=\frac{1}{Z_\La (z, \beta )}
\int_{\Ga_\La}e^{-\be
U(\eta\cup\ga)}\la_{z\si}(d\ga),\;\;\;\eta\in\Ga_{\La},
\end{equation}
and corresponding correlation functions of quasi-continuous
approximation $\rho_\La^{(-)}(\eta; z, \beta, a)$ are defined as:
\begin{equation} \label{28}
\rho_\La^{(-)}(\eta; z, \beta, a) :=\frac{1}{Z_\La^{(-)}(z, \beta,
a)} \int_{\Ga_\La}e^{-\be U(\eta\cup\ga)}\la^a_{z\si}(\eta\cup
d\ga),
\end{equation}
where  according to \eqref{24}
\begin{equation} \label{24-1}
\la^a_{z\si}(\eta\cup d\ga):
=\prod_{\Delta\in\overline{\Delta_{a}}\cap\Lambda}\chi^\Delta_-(\eta\cup\gamma)\la_{z\si}(d\ga).
\end{equation}

\section{\bf   Main results } We prove the results for the infinite
volume characteristics, so let
 $(\La_{l})$ be a sequence of bounded Lebesgue measurable regions of
$\R^d$:
\begin{equation}\label{29}
\La_1\subset\La_2\subset\ldots\subset\La_n\subset\ldots,\;\;\underset{l}\cup\La_l\;=\;\R^d,
\end{equation}
and the sequence $(\La_{l})$ tends to $\R^d$ in the sense of
Fisher (see \cite{Ru69}, Ch.2, S. 2.1).

 It is well-known that for any configuration $\eta\in\Ga_0$ and any sequence
\eqref{29}, such that $\eta\subset\La_1$ there exists subsequence
$(\La^{'}_{k})$ of $(\La_{l})$, such that
\begin{equation}\label{30}
\lim_{k\rightarrow \infty}\rho_{\La^{'}_{k}}(\eta;z,\beta)=
\rho(\eta;z,\beta) <\infty
\end{equation}
 for all positive $z,\beta$ uniformly on $\frak{B}(\Ga_0)$ . This result follows from the uniform
bounds of the family $\{\rho_{\La_l}:\La_l\in\mathcal{B}_c({{\Bbb
R}}^{d})$:
\begin{equation}\label{30-1}
\rho_{\La_{l}}(\eta;z,\beta)\leq\;\xi^{|\eta|}e^{-\beta\,U_{\delta}^{+}}
\end{equation}
with some positive $\xi$, independent of $\La_l$, $\eta$.

The inequality \eqref{30-1} without exponent in r.h.s. was
obtained for the first time in the article \cite{Ru70}. In the
work \cite{Re98} a new proof (much easier) was presented with
exponent $e^{-\beta\,U_{1/2}^{+}}$ and in the articles \cite{KR04}
and \cite{PR07} it was proved for many-body interactions for
finite range and infinite range cases respectively.

In the next section we  give a sketch of proof of the following
lemma

\begin{lem}\label{4.1}
Let the interaction potential $V$  satisfy the assumptions
\textbf{(A)} for two-body  and ${\bf A1-A5}$ for many-body
interactions. Then there exist some $0<a_*\leq a_0<r_0$ and a
positive constant $\xi_{-}=\xi_{-}(a_*)$, which does not depend on
$\La_l,\,a$ and $\eta$, s.t.
\begin{equation}\label{30-6}
\rho^{(-)}_{\La_l}(\eta; z, \beta,
a)\leq\xi_{-}^{|\eta|}e^{-\beta\,U_{\delta}^{+}},
\end{equation}
holds for any  $a<a_*$ such that $a_*/a\in\N$.
\end{lem}

 So, as in the previous case, there
exists subsequence ($\La_m^{\prime\prime}$) of the sequence
($\La_l$) such that one can define
\begin{equation}\label{30-7}
\rho^{(-)}(\eta;z,\beta, a)\;=\;\lim_{m\rightarrow
\infty}\rho_{\La^{\prime\prime}_{m}}^{(-)}(\eta;z,\beta, a)
 <\infty.
\end{equation}

\begin{rem}\label{r21}
The limit functions $\rho(\eta;z,\beta)$ and
$\rho^{(-)}(\eta;z,\beta, a)$ in \eqref{30} and \eqref{30-7} can
be different for different subsequences $\La^{\prime}_k$ and
$\La^{\prime\prime}_m$. So, in order  to make the function
$\rho^{(-)}(\eta;z,\beta, a)$  be  approximation of the function
$\rho(\eta;z,\beta)$ we have to take the subsequence
$\La^{\prime\prime}_m$ in the limit \eqref{30-7} as some
subsequence $\La^{\prime}_k$.
\end{rem}

Then we can formulate the following result.

\begin{theorem}\label{th1}  Let the interaction potential $V$  satisfy the assumptions
\textbf{(A)} for two-body  and ${\bf A1-A5}$ for many-body
interactions. Then for any $\varepsilon
> 0$, any positive  $z$ and $\be$ and any configuration $\eta\in
\Gamma_{0}$  there exists \, $a = a(z, \beta, \varepsilon)
> 0$ such that:
\begin{equation}\label{30-2}
|\rho(\eta;z,\beta) - \rho^{(-)}(\eta;z,\beta, a)| < \varepsilon,
\end{equation}
where $\rho(\eta;z,\beta)$ and $\rho^{(-)}(\eta;z,\beta, a)$ are
the limits of $\rho_{\La^{\prime\prime}_{m}}(\eta;z,\beta)$ and
$\rho_{\La^{\prime\prime}_{m}}^{(-)}(\eta;z,\beta, a)$
respectively with the same subsequence of the sequence $(\La_l)$
(see Remark \ref{r21}).
\end{theorem}

\begin{proof}

The proof is based on the  existence of the limits \eqref{30},
\eqref{30-7} and the following  lemma.

\begin{lem}\label{l3.3}
Let the interaction potential $V$  satisfy the assumptions
\textbf{(A)} for two-body  and ${\bf A1-A5}$ for many-body
interactions. Then for any sequence $\La_l$ of the type \eqref{29}

\begin{equation}\label{30-9}
\lim_{a\rightarrow 0}\rho^{(-)}_{\La_{l}}(\eta; z, \beta,
a)=\rho_{\La_{l}}(\eta; z, \beta).
\end{equation}
and hence for any $\varepsilon>0$ there exists $a<a_*$, s.t. the
following inequality holds:
\begin{equation}\label{30-10}
|\rho^{(-)}_{\La_{l}}(\eta; z, \beta, a)-\rho_{\La_{l}}(\eta; z,
\beta)|\leq\frac{\varepsilon}{3}.
\end{equation}
\end{lem}

From the existence of the limits \eqref{30} and \eqref{30-1} for
any $\varepsilon>0$ $\exists K_1\in\N$, s.t. for any $k\geq K_1$
the following inequality holds:

\begin{equation}\label{30-11}
|\rho_{\La^{\prime\prime}_{m}}(\eta; z, \beta)-\rho(\eta;
z,\beta)|\leq\frac{\varepsilon}{3}.
\end{equation}

 and  $\exists K_2\in\N$,
s.t. for any $k\geq K_2$ the following inequality holds:
\begin{equation}\label{30-12}
|\rho^{(-)}_{\La^{\prime\prime}_{m}}(\eta; z, \beta,
a)-\rho^{(-)}(\eta; z, \beta, a)|\leq\frac{\varepsilon}{3}.
\end{equation}

Then the statement of the theorem {\bf 4.1} follows from
\eqref{30-10} with $\La_l\equiv\La^{\prime\prime}_{m}$ and
\eqref{30-11}, \eqref{30-12}:
\begin{align*}
&|\rho(\eta; z, \beta)-\rho^{(-)}(\eta; z, \beta, a)|=\\
&=|\rho(\eta; z, \beta)-\rho_{\La^{\prime\prime}_{m}}(\eta; z,\beta)+ \\
&+\rho_{\La^{\prime\prime}_{m}}(\eta; z,\beta)-\rho^{(-)}_{\La^{\prime\prime}_{m}}(\eta;z, \beta, a)+\\
&+\rho^{(-)}_{\La^{\prime\prime}_{m}}(\eta; z, \beta, a)-\rho^{(-)}(\eta; z, \beta, a)|\leq\\
&\leq\frac{\varepsilon}{3}+\frac{\varepsilon}{3}+\frac{\varepsilon}{3}=\varepsilon
\end{align*}

\end{proof}

 \begin{corollary}\label{co1}
  The inequality   \eqref {30-2} ensures   existence of
 the limit:
\begin{equation}\label{30-3}
\lim_{a\rightarrow 0}\rho^{(-)}(\eta;z,\beta,a)=
\rho(\eta;z,\beta).
\end{equation}
 for any  positive $z, \;\textnormal{any} \beta>0$ and  $\eta\in\Gamma_0.$
\end{corollary}
For two-body interaction this result in the region of sufficiently
small values of a parameter $z$ is  obtained in the article
\cite{RT09}.

\section{\bf   Proof of the Lemmas 4.1 and 4.2}

\subsection{\bf   Proof of the Lemma 4.1}

 {\it The proof of the lemma 4.1}  is based
on the expansion of correlation functions into  dense
configurations which was proposed in \cite{PR07}
 (see, also,  \cite{PR09}) and actually coincides with the proof of the theorem 2.2
 of the article \cite{PR07} for two-body interaction and with the
 proof of the theorem 2.1 of the article  \cite{PR09} for many-body
 interaction. The main difference in proving the lemma 4.1  is that in
 the definition of the correlation functions $\rho_\La^{(-)}(\eta;z,\beta,a)$
the integrals are w.r.t. the measure $\la^a$ (see \eqref{28},
\eqref{24-1}), unlike in is in  the definition of the correlation
functions $\rho_\Lambda(\eta; z, \beta)$ where  the integrals are
w.r.t. the measure $\la$ (see \eqref{23}--\eqref{27}) which takes
into account all possible configurations. So, the main goal of
this lemma is to show that the constant $\xi_{-}$ in the
inequality \eqref{30-6} does not depend on the parameter $a$. So,
in this section we give only main point in the construction of
expansion and estimate some value which did not appear in the
previous proofs.

 In order to arrange this
expansion let us define also an indicator of a dense configuration
in any cube $\Delta\in\overline{\Delta}_{a}$ as
$\chi_{+}^{\Delta}(\gamma) \;= \; 1-\chi_{-}^{\Delta}(\gamma)$.

Then we use the following partition of the unity for any
$\gamma\in\Gamma_\Lambda$ with $a= a_*$, i.e.
$\overline\De_{a_*}$:

\begin{align}\label{41}
 1 \;= \; \prod_{\Delta\subset \Lambda(a_*)} \left[
\chi_{-}^{\Delta}( \gamma)+\chi_{+}^{\Delta}( \gamma) \right]
 \;&= \; \sum_{n=0}^{N_{\Lambda(a_*)}}\sum_{\{\Delta_1,...,\Delta_n\}\subset\Lambda(a_*)}
 \prod_{i=1}^{n}\chi_{+}^{\Delta_{i}}(\gamma)\prod_{\Delta\subset\Lambda(a_*)\setminus\cup_{i=1}^n\Delta_i}\chi_{-}^{\Delta}(\gamma)\;:=\notag\\
 &:=\;\sum_{\emptyset\subseteq X\subseteq\Lambda
(a_*)}\widetilde\chi_+^{X}(\gamma) \widetilde\chi_-^{\Lambda
(a_*)\setminus X}(\gamma),
 \end{align}
where $N_\Lambda\;=\;|\Lambda|/a_*^d$ (here the symbol $|\La|$
means Lebesgue measure of the set $\La(a_*)$) is the number of
cubes $\Delta$ in the volume $\Lambda=\La(a_*)$(see subsection
2.4), and

\begin{equation}\label{42}
\widetilde\chi_{\pm}^{X}(\gamma) \;= \; \prod_{\Delta\subset
X}\chi_{\pm }^{\Delta}(\gamma).
 \end{equation}
Inserting \eqref{41} with  $a= a_*$ into the definition \eqref{28}
of correlation functions $\rho_\Lambda^{(-)}(\eta; z, \beta, a)$
with $a< a_*$ s.t. $\frac{a_*}{a} \in \mathbb{N}$ we obtain:
\begin{equation}\label{43}
\rho_\Lambda^{(-)}(\eta; z, \beta, a) \;=
 \; \frac{1}{Z_\La^{(-)}(z, \beta, a)}\sum_{\emptyset\subseteq X\subseteq\Lambda(a_*)}\quad
\int_{\Ga_\La}
   e^{-\beta U(\eta\cup\ga)}\,\widetilde\chi_+^{X }(\ga)
\widetilde\chi_-^{\Lambda\setminus X}(\ga)\la^a_{z\si}(\eta\cup
d\ga).
\end{equation}

\begin{rem}\label{r41}
We want to stress that the sets $X$ in  \eqref{43} are the unions
of cubes $\De\in\overline{\Delta}_{a_*}$, but in the product of
the definition $\la^a_{z\si}(\eta\cup d\ga)$ (see \eqref{24-1})
$\De\in\overline{\Delta}_{a}$ with $a<a_*$ and $a_*/a\in\N$.
\end{rem}

The next steps in the construction of expansion and estimates are
completely the same as in the proof of the theorem 2.2 of the
article \cite{PR07} for two-body interaction and the theorem 2.1
of the article \cite{PR09} for many-body interaction. It is
necessary only to note that to change the integration w.r.t
measure $\la^a_{z\si}(\eta\cup d\ga_{\De'})$  for the integration
w.r.t. measure $\la^a_{z\si}(d\ga_{\De'})$ (see \eqref{24-1}) we
use the following inequality:

\[
\chi^{\Delta'}_-(\eta\cup\gamma_{\Delta'}) \leq
\chi^{\Delta'}_-(\gamma_{\Delta'}),
\]
which follows from the definition \eqref{22} for any
$\De'\in\overline{\Delta_{a}}$ and any $\ga\in\Ga$. \hfill
$\blacksquare$

\subsection{\bf   Proof of the Lemmas  4.2}

Let us insert now the unity \eqref{41} (but with partition $\La$
into cubes with edges $a$ instead of $a_*$ and the argument
$\eta\cup\ga$ in each function $\chi_{\pm}^\De$ )  in \eqref{27}.
Then we obtain the following expansion:

\begin{align}\label{57}
\rho_\Lambda(\eta;z,\beta) \;=
 \; \frac{z^{|\eta|}}{Z_\La(z, \beta)}
\sum_{X\subseteq\Lambda}\quad \int_{\Ga_\La}
   e^{-\beta U(\eta\cup\ga)}\,\widetilde\chi_+^{X}(\eta\cup\ga)
\widetilde\chi_-^{\Lambda\setminus
X}(\eta\cup\ga)\la_{z\si}(d\ga).
\end{align}

Extracting the first term at $X=\emptyset$ and using the
definitions \eqref{23}-\eqref{28} we can rewrite \eqref{57} in the
following form:
\begin{equation}\label{58}
\rho_{\Lambda}(\eta;z, \beta) \;= \;\frac{Z_\La^{(-)}(z, \beta, a
)}{Z_\La(z, \beta) }\rho_{\Lambda}^{(-)}(\eta;z, \beta, a)
 + R^{\Lambda}(\eta;z,\beta, a),
\end{equation}
where
\begin{align}\label{59}
R^{\Lambda}(\eta;z,\beta, a)\;=
 \; \frac{z^{|\eta|}}{Z_\La(z, \beta)}
 \sum_{\emptyset\ne X\subseteq\Lambda}\quad
\int_{\Ga_\La}
   e^{-\beta U(\eta\cup\ga)}\,\widetilde\chi_+^{X}(\eta\cup\ga)
\widetilde\chi_-^{\Lambda\setminus
X}(\eta\cup\ga)\la_{z\si}(d\ga).
\end{align}

The proof of the lemma 4.2  is based on  two  technical lemmas.

\begin{lem}\label{l.4}
Let the interaction potential $V$  satisfy the assumptions
\textbf{(A)} for two-body  and ${\bf A1-A5}$ for many-body
interactions. Then for any fixed volume $\La \in\frak{B}_c(\R^d)$
and any configuration $\eta \in \Gamma_0$ the following holds:
\begin{align}\label{60}
&\underset{a\rightarrow
0}{\textnormal{lim}}R^{\Lambda}(\eta;z,\beta, a)=0.
\end{align}
\end{lem}

{\it Proof.} See  Appendix.\hfill $\blacksquare$
\begin{lem}\label{l2}
Let the interaction potential $V$  satisfy the assumptions
\textbf{(A)} for two-body  and ${\bf A1-A5}$ for many-body
interactions.  Then for any fixed volume $ \La\in\frak{B}_c(\R^d)$
the following holds:
\begin{equation}\label{61}
\underset{a\rightarrow
0}{\textnormal{lim}}\frac{Z_\Lambda^{(-)}(z, \beta,
a)}{Z_\Lambda(z, \beta)}=1.
\end{equation}
\end{lem}
{\it Proof.} In the articles  \cite{RT07} and \cite{Pe08}  the
following estimate was  obtained:
\[\underset{a\rightarrow
0}{\textnormal{lim}}\;\frac{Z_\Lambda^{(-)}(z, \beta,
a)}{Z_\Lambda(z, \beta)}\geq 1,\] on which the proof of the fact
that the pressure of approximated system converges to the
pressure of the real system is based. From the other hand in
accordance with the definitions \eqref{23}, \eqref{24} it is clear
that

\[\frac{Z_\Lambda^{(-)}(z, \beta, a)}{Z_\Lambda(z, \beta)}\leq 1.\]
As a result we have

\[\underset{a\rightarrow
0}{\textnormal{lim}}\frac{Z_\Lambda^{(-)}(z, \beta,
a)}{Z_\Lambda(z, \beta)}=1.\]

\section{\bf   Appendix }

{\it Proof of the  lemma 5.1}

Using \eqref{15} for two-body potential and \eqref{63} for
many-body interaction, one can rewrite \eqref{59} in such a  way:
\begin{align}\label{A3}
 R^{\Lambda}(\eta;z,\beta, a)&= \;\frac{z^{|\eta|}}{Z_\La(z,\beta)}
\underset{\emptyset \neq X \subseteq
\Lambda}{\sum}\int_{\Ga_\La}e^{-\beta (\frac{1}{2}\w{U}^+(\eta\cup
 \gamma_X)+U^{st}(\eta\cup
 \gamma_X))}\widetilde\chi_+^{X}(\eta\cup\ga)\times \notag \\
 &e^{-\beta W(\eta \cup \gamma_X; \gamma_{\Lambda\setminus
 X})-\frac{1}{2}\beta\w{U}^+(\eta\cup
 \gamma_X)}e^{-\beta U(\gamma_{\Lambda\setminus X})}\widetilde\chi_-^{\Lambda \setminus X}(\eta\cup\ga)\la_{z\si}(d\ga).
\end{align}

Using infinite divisibility property of Lebesgue-Poisson  measure,
 the   estimate: \[ e^{-\beta W(\eta \cup \gamma_X; \gamma_{\Lambda\setminus
 X})-\frac{1}{2}\beta\w{U}^+(\eta\cup
 \gamma_X)}\leq e^{\beta\upsilon_*(a)(|\eta|+|\gamma_X|)}\]

  and the fact that\[\widetilde \chi_-^{\Lambda \setminus X}(\eta \cup\ga)\leq 1,\]
where $\upsilon_*(a)=\upsilon_0(a)$ for two-body potential and
 $\upsilon_*(a)=\bar{I}(a)$ for many-body interaction (see.
 \eqref{lem3-1}),
   we obtain from
 \eqref{A3}:
\begin{align}\label{A4}
 R^{\Lambda}(\eta;z,\beta, a)&\leq \;\frac{({z e^{\beta\upsilon_*(a)})^{|\eta|}}}{Z_\La(z,\beta)}
 \underset{\emptyset \neq X \subseteq \Lambda}{\sum}\int_{\Ga_X}e^{-\beta (\frac{1}{2}\w{U}^+(\eta\cup
 \gamma_X)+U^{st}(\eta\cup
 \gamma_X)+\upsilon_*(a)|\gamma_X|)} \times \notag\\
 &\times\widetilde\chi_+^{X}(\eta\cup\ga)\la_{z\si}(d\ga_{ X})
 \int_{\Gamma_{\Lambda\setminus X}}e^{-\beta U(\gamma_{\Lambda\setminus X})}\la_{z\si}(d\ga_{\Lambda\setminus X}).
\end{align}
Let us take into account that \[Z_{\Lambda\setminus X }(z, \beta
)=\int_{\Gamma_{\Lambda\setminus X}}e^{-\beta
U(\gamma_{\Lambda\setminus X})}\la_{z\si}(d\ga_{\Lambda\setminus
X})\]
 and
$Z_{\Lambda\setminus X }(z, \beta )\leq Z_{\Lambda }(z, \beta )$.

Then we have from \eqref{A4}:
\begin{equation}\label{A5}
 R^{\Lambda}(\eta;z,\beta, a)\leq \;({z
 e^{\beta\upsilon_*(a)})^{|\eta|}}
 \underset{\emptyset \neq X \subseteq \Lambda}{\sum}\int_{\Ga_X}e^{-\beta (\frac{1}{2}\w{U}^+(\eta\cup
 \gamma_X)+U^{st}(\eta\cup
 \gamma_X)+\upsilon_*(a)|\gamma_X|)}\widetilde\chi_+^{X}(\eta\cup\ga)\la_{z\si}(d\ga_{
 X}).
\end{equation}
Let $\La_\eta$ be a union of cubs which contain points from the
configuration $\eta$. Then using Proposition 3.1, lemma 3.1 and
inequalities \eqref{9},\eqref{9-1} we have:
\begin{equation}\label{A6}
 R^{\Lambda}(\eta;z,\beta, a)\leq \;({z
 e^{\beta(\upsilon_*(a)+B(a))})^{|\eta|}}( R_1^\Lambda+R_2^\Lambda),
\end{equation}
where
\begin{equation*}\label{A61}
R_1^\Lambda= \underset{\emptyset \neq X \subseteq (\Lambda
\setminus\Lambda_\eta)}{\sum}\int_{\Ga_X}e^{ \underset{\Delta\in
X} {\sum}\; \beta \left(-\frac{1}{2}A(a)|\gamma_{\Delta}|^2 +
(B(a)+\upsilon_*(a)) |\gamma_\Delta|\right)}
\widetilde\chi_+^{X}(\eta\cup\ga)\la_{z\si}(d\ga_{
 X}),
\end{equation*}
\begin{equation*}\label{A62}
R_2^\Lambda=  \underset{ \substack{\emptyset \neq X \subseteq
\Lambda,\\X\cap\Lambda_\eta\neq \emptyset} }{\sum}\int_{\Ga_X}e^{
\underset{\Delta\in (X } {\sum}\; \beta
\left(-\frac{1}{2}A(a)(|\gamma_{\Delta}|+|\eta_\Delta|)^2 +
(B(a)+\upsilon_*(a)) |\gamma_\Delta|\right)}
\widetilde\chi_+^{X}(\eta\cup\ga)\la_{z\si}(d\ga_{
 X})
\end{equation*}
with $A(a)$ and $B(a)$ as in \eqref{218} for two-body potentials
and \eqref{9-1} for many-body interaction. Using  again the
infinite divisible property of the Lebesgue-Poisson measure and
its definition one can calculate
\begin{align}\label{39-1}
 &\int_{\Gamma_{\Delta}}\,
e^{-\beta\frac{1}{2}\,A(a)\,|\gamma_\Delta|^2+\beta\,(B(a)+\upsilon_*(a)
)
\,|\gamma_\Delta|}\,\chi_+^\Delta(\gamma_\Delta)\,\lambda_{z\sigma}(d\gamma_{\Delta})=\\
&=\;\sum\limits_{n=2}^{\infty}\,\frac{(a^d\,z)^n}{n!}\,
e^{-\frac{1}{2}\beta\,A(a)\,n^2+\beta\,(B(a)+\upsilon_*(a))\,n}\,\leq\,\epsilon_1
(a),\notag
\end{align}
where
\begin{equation}\label{39}
 \epsilon_1(a)\;\;\rightarrow\;\; 0,\;\;\text{якщо}\;\; a\;\rightarrow\;\;0.
\end{equation}
Then after summing w.r.t. $X$ we obtain the following estimate:
\begin{equation}\label{A63}
R_1^\Lambda \leq \;
(1+\epsilon_1(a))^{\frac{|\Lambda\setminus\Lambda_\eta|}{a^d}}-1\leq
\epsilon_1(a)\frac{|\Lambda\setminus\Lambda_\eta|}{a^d}
(1+\epsilon_1(a))^{\frac{|\Lambda\setminus\Lambda_\eta|}{a^d}-1}.
\end{equation}
To estimate $R_2^\La$ let us rewrite it in the form:

\begin{align}\label{A7}
 R_2^\Lambda =\;&
 \underset{ \substack{\emptyset \neq X \subseteq
\Lambda,\\X\cap\Lambda_\eta\neq \emptyset}
}{\sum}R_0^\Lambda(\eta_{X \cap \Lambda_\eta};z,\beta, a)\times\\
 &\times\int_{\Ga_{X \setminus \Lambda_\eta}}e^{
 \underset{\Delta\subset
(X \setminus \Lambda_\eta)} {\sum}\; \beta
\left(-\frac{1}{2}A(a)|\gamma_{\Delta}|^2 + (B(a)+\upsilon_*(a))
|\gamma_\Delta|\right)}\widetilde\chi_+^{X
\setminus\Lambda_\eta}(\ga_{X\setminus \Lambda_\eta
})\la_{z\si}(d\ga_{
 X\setminus\Lambda_\eta}),\notag
\end{align}
where
\begin{align}\label{A8}
  &R_0^\Lambda=\int_{\Ga_{X\cap \Lambda_\eta}}e^{
  \underset{\Delta\subset
 X \cap \Lambda_\eta } {\sum}\;\beta
\left(-\frac{1}{2}A(a)(|\eta_{\Delta}|+|\gamma_{\Delta}|)^2 +
(B(a)+\upsilon_*(a)) |\gamma_\Delta|\right)}
\widetilde\chi_+^{X\cap \Lambda_\eta}(\eta\cup\ga_{X\cap
\Lambda_\eta })\la_{z\si}(d\ga_{
 X\cap \Lambda_\eta})=\notag \\
&= \underset{\Delta \in X \cap
\Lambda_\eta}{\prod}\int_{\Gamma_\Delta}e^{\beta
\left(-\frac{1}{2}A(a) (|\eta_{\Delta}|+|\gamma_{\Delta}|)^2 +
(B(a)+\upsilon_*(a)) |\gamma_\Delta|\right)}
\chi_+^{\Delta}(\eta_\Delta\cup\gamma_\Delta)\la_{z\si}(d\ga_{
 \Delta}).
\end{align}
Estimating maximum of the exponent we obtain:
\begin{align}\label{A9}
R_0^\Lambda(\eta_{X\cap\La_\eta};z,\beta, a)\leq
e^{-\beta\left(2A(a)-B(a)-\upsilon_*(a)\right)}\underset{\Delta
\subset
 X \cap
\Lambda_\eta}{\prod}\int_{\Gamma_\Delta} \la_{z\si}(d\ga_{
 \Delta})\leq e^{-\beta\left(2A(a)-B(a)-\upsilon_*(a)\right)}e^{z a^d |\eta|}
\end{align}

Using \eqref{A7}, \eqref{A9}  we can  estimate $R_2^\Lambda$ from
above in the form:
\begin{align}\label{A10}
 &R_2^{\Lambda}\leq e^{-\beta\left(2A(a)-B(a)-\upsilon_*(a)\right)}e^{z a^d
 |\eta|}\times \notag \\
 &\underset{ \substack{\emptyset \neq X \subseteq
 \Lambda,\\X\cap\Lambda_\eta\neq \emptyset}
}{\sum}\;\;\;\int_{\Ga_{X \setminus \Lambda_\eta}}e^{
\underset{\Delta\subset (X \setminus \Lambda_\eta) } {\sum}\;
\beta \left(-\frac{1}{2}A(a)|\gamma_{\Delta}|^2 +
(B(a)+\upsilon_*(a)) |\gamma_\Delta|\right)} \widetilde\chi_+^{X
\setminus\Lambda_\eta}(\ga_{X\setminus \Lambda_\eta
})\la_{z\si}(d\ga_{
 X\setminus\Lambda_\eta}).
\end{align}
Let us take into account that for any
$\frak{B}(\Ga_\La)$-measurable function $F(\gamma)$ the following
holds:
 \begin{align}
&\underset{ \substack{\emptyset \neq X \subseteq
\overline{\Delta}_a\cap \Lambda,\\X\cap\Lambda_\eta\neq \emptyset}
}{\sum}\;\;\;\int_{\Ga_{X \setminus
\Lambda_\eta}}F(\gamma_{X\setminus \Lambda_\eta})\la_{z\si}(d\ga_{
 X\setminus \Lambda_\eta})\leq \notag \\
 &(2^{|\eta|}-1)
\underset{ X \subseteq \overline{\Delta}_a\cap \Lambda \setminus
\Lambda_\eta  }{\sum}\;\;\;\int_{\Ga_{X
}}F(\gamma_X)\la_{z\si}(d\ga_{
 X})
\end{align}
Using this fact and infinite divisibility property of
Lebesgue-Poisson measure we obtain from \eqref{A10}:
\begin{align}\label{A11}
 &R_2^{\Lambda}\leq e^{-\beta\left(2A(a)-B(a)-\upsilon_*(a)\right)}e^{a^d |\eta|}
 (2^{|\eta|}-1)
\times \notag \\
&\underset{ X \subseteq  \Lambda \setminus \Lambda_\eta
}{\sum}\;\underset{\Delta \subset
X}{\prod}\;\int_{\Gamma_\Delta}e^{\beta \left(-\frac{1}{2}A(a)
|\gamma_{\Delta}|^2 + (B(a)+\upsilon_*(a))
|\gamma_\Delta|\right)}\times \notag
\\
&\chi_+^{\Delta}(\ga_{\Delta})\la_{z\si}(d\ga_{
 \Delta})\leq e^{-\beta\left(2A(a)-B(a)-\upsilon_*(a)\right)}e^{z a^d |\eta|}
 (2^{|\eta|}-1) \left(1+\epsilon_1(a) \right)^{\frac{|\Lambda\setminus\Lambda_\eta|}{a^d}}.
\end{align}
It follows from \eqref{A6}, \eqref{A63}, \eqref{A11} that:
\begin{align}\label{A12}
&R^{\Lambda}(\eta;z,\beta, a)\leq (z e^{\beta
(B(a)+\upsilon_*(a))})^{|\eta|}\left(1+\epsilon_1(a)\right)^{\frac{|\Lambda\setminus
\Lambda_\eta
|}{a^d}-1}\biggl(\epsilon_1(a)\frac{|\Lambda\setminus\Lambda_\eta|}{a^d}+\notag \\
&(2^{|\eta|}-1)(1+\epsilon_1(a))e^{-\beta\left(2A(a)-B(a)-\upsilon_*(a)\right)}e^{z
a^d |\eta|}\biggl)\rightarrow 0, \textnormal{якщо}\; a \rightarrow
0,
\end{align}

This is the end of the proof.

$$
\mspace{675mu}\blacksquare
$$

\end{document}